\documentclass[11pt]{article}

\newcommand{\jpa}{J. Phys. A~}

\newcommand{\natphy}{Nature Phys.~}

\newcommand{\prl}{Phys. Rev. Lett.~}
\newcommand{\pra}{Phys. Rev. A~}

\newcommand{\pla}{Phys. Lett. A~}

\usepackage[latin1]{inputenc}
\usepackage[linesnumbered,boxed,ruled,commentsnumbered]{algorithm2e}
\usepackage{appendix}
\usepackage{epstopdf}
\usepackage{xcolor}
\usepackage{subfigure}
\usepackage[T1]{fontenc}
\usepackage[sc]{mathpazo}
\usepackage{amsmath}
\usepackage{amssymb}
\usepackage{graphicx,color}
\usepackage{framed}
\usepackage{multirow}
\usepackage{enumerate}
\usepackage{amsthm}
\usepackage{amsfonts,mathrsfs}
\usepackage{geometry} 
\usepackage{eepic}
\usepackage{ifthen}
\usepackage[vcentermath]{youngtab}
\usepackage[unicode=true,pdfusetitle, bookmarks=true,bookmarksnumbered=false,bookmarksopen=false, breaklinks=false,pdfborder={0 0 0},backref=false,colorlinks=false] {hyperref}
\hypersetup{
colorlinks,linkcolor=myurlcolor,citecolor=myurlcolor,urlcolor=myurlcolor}
\definecolor{myurlcolor}{rgb}{0,0,0.7}

\usepackage{color}

\newcommand{\blue}{\textcolor{blue}}

\usepackage{hyperref}
\hypersetup{pdfpagemode=UseNone}

\newcommand{\tinyspace}{\mspace{1mu}}

\newcommand{\abs}[1]{\left\lvert\tinyspace #1 \tinyspace\right\rvert}

\renewcommand{\det}{\operatorname{det}}

\newcommand{\setft}[1]{\mathrm{#1}}

\newcommand{\density}[1]{\setft{D}\left(#1\right)}

\def \dif {\mathrm{d}}

\def\complex{\mathbb{C}}
\def\real{\mathbb{R}}

\newenvironment{mylist}[1]{\begin{list}{}{
    \setlength{\leftmargin}{#1}
    \setlength{\rightmargin}{0mm}
    \setlength{\labelsep}{2mm}
    \setlength{\labelwidth}{8mm}
    \setlength{\itemsep}{0mm}}}
    {\end{list}}

\def\ot{\otimes}

\newcommand{\out}[2]{| #1\rangle\langle #2 |}

\newcommand{\Pa}[1]{\left(#1\right)}

\newcommand{\Br}[1]{\left[#1\right]}
\newcommand{\set}[1]{\{#1\}}
\newcommand{\Set}[1]{\left\{#1\right\}}

\newcommand{\ket}[1]{|#1\rangle}

\DeclareMathOperator{\trace}{Tr}

\newcommand{\Ptr}[2]{\trace_{#1}\Pa{#2}}

\newcommand{\Tr}[1]{\Ptr{}{#1}}

\newcommand{\Abs}[1]{\left|\tinyspace#1\tinyspace\right|}

\def\cA{\mathcal{A}}

\def\cN{\mathcal{N}}

\def\cV{\mathcal{V}}

\def\bsU{\boldsymbol{U}}

\def\bsp{\boldsymbol{p}}

\newtheorem{thrm}{Theorem}[section]
\newtheorem{lem}[thrm]{Lemma}
\newtheorem{prop}[thrm]{Proposition}

\theoremstyle{definition}

\newtheorem{remark}[thrm]{Remark}
\newtheorem{exam}[thrm]{Example}

\numberwithin{equation}{section}

\newcounter{questionnumber}

\begin{document}

\title{\Large \bf Visualization of all two-qubit states via partial-transpose-moments}

\author{\blue{Lin Zhang}$^1$\footnote{godyalin@163.com},\quad \blue{Yi Shen}$^2$\footnote{yishen@jiangnan.edu.cn (corresponding author)},\quad \blue{Hua Xiang}$^3$,\quad \blue{Quan Qian}$^1$,\quad \blue{Bo Li}$^1$\\
  {\it\small $^1$School of Science, Hangzhou Dianzi University, Hangzhou 310018, PR~China}\\
  {\it\small $^2$School of Science, Jiangnan University, Wuxi 214122, PR~China} \\
  {\it\small $^3$School of Mathematics and Statistics, Wuhan University, Wuhan 430072, PR~China}}
\date{\today}
\maketitle

\begin{abstract}
Efficiently detecting entanglement based on measurable quantities is
a basic problem for quantum information processing. Recently, the
measurable quantities called partial-transpose (PT)-moments have
been proposed to detect and characterize entanglement. In the
recently published paper [L. Zhang \emph{et al.},
\href{https://doi.org/10.1002/andp.202200289}{Ann. Phys.(Berlin)
\textbf{534}, 2200289 (2022)}], we have already identified the
2-dimensional (2D) region, comprised of the second and third
PT-moments, corresponding to two-qubit entangled states, and
described the whole region for all two-qubit states. In the present
paper, we visualize the 3D region corresponding to all two-qubit
states by further involving the fourth PT-moment (the last one for two-qubit
states). The characterization of this 3D region can finally be
achieved by optimizing some polynomials. Furthermore, we identify
the dividing surface which separates the two parts of the whole 3D region
corresponding to entangled and separable states respectively. Due to
the measurability of PT-moments, we obtain a complete and operational criterion
for the detection of two-qubit entanglement.
\end{abstract}


\section{Introduction}
\label{sec:intro}

Quantum entanglement has been regarded as an indispensable resource
in various information processing tasks. Nevertheless, the detection
and characterization of entanglement is not easy, especially for
high-dimensional and multipartite systems. A well-known fact is that
to determine a bipartite state is separable or entangled is an
NP-hard problem \cite{Gurvits2003,Gurvits2004,Gharibian2008}. For
this notorious difficulty, the low-dimensional systems, especially
the two-qubit system, are still the main setting in quantum
experiments \cite{Zeilinger1997}. Fortunately, the celebrated PPT
(positive-partial-transpose) criterion is sufficient and necessary
for detecting bipartite entanglement in two-qubit and qubit-qutrit
systems \cite{Peres1996,Horodecki1996}. In this paper, we focus on
two-qubit states, and characterize both the set of entangled states
and that of separable states from a graphical point of view which offers an effective tool to understand entanglement
\cite{Bengtsson2017}.

The PPT criterion is actually a mathematical tool commonly used in
the theoretical scenario because the partial transposition is not a
physical operation. Theoretically, for a given bipartite state
$\rho_{AB}$ of system $AB$, detecting its entanglement by PPT
criterion is straightforward to check whether the partial transpose
of $\rho_{AB}$, denoted by $\rho_{AB}^{\Gamma}$, is still positive
semidefinite. Nevertheless, in the practical scenario, the detection
of entanglement is not so direct. On the one hand, the partial
transposition of an operator cannot be implemented directly in
experiments. On the other hand, in actual experiments, the quantum
state is unknown, unless the resource-inefficient quantum state
tomography is performed \cite{Yu2021}. More generally, the
realizations of effective entanglement criteria usually consume
exponential resources, and efficient criteria often perform poorly
without prior knowledge \cite{Liu2022aug}. Therefore, it is more
valuable to develop efficient methods of detecting entanglement
based on measurable quantities from the perspective of resource
conservation.

The so-called \emph{partial transpose moments} (PT-moments)
introduced in \cite{Yu2021} has recently been recognized as the good
measurable quantities to efficiently detect entanglement. The $k$-th
PT-moment, denoted by $p_k$, of a state $\rho_{AB}$ is defined as
the sum of $k$th-power of each eigenvalue of $\rho_{AB}^{\Gamma}$
for any positive integer $k$ not greater than the global dimension
of the bipartite system $AB$. Although the PT-moments are nonlinear
functions of the state $\rho_{AB}$ depending on the spectrum of
$\rho^{\Gamma}_{AB}$, it has been shown that the PT-moments can be
efficiently measured from randomized measurements which are readily available in NISQ devices \cite{Elben2020,Huang2020}. The protocol of determining PT-moments proposed in Ref. \cite{Elben2020} is efficient for the following three facts. First, compared with previous proposals, this protocol only requires single-qubit control, and allows for the estimation of many distinct PT-moments from the same data \cite{Elben2020}. Second, compared with performing full state tomography to compute PT-moments with high accuracy, this protocol requires less measurements, and the number of required measurements decreases significantly for highly mixed states \cite{Elben2020}. Third, the PT-moments are predicted by classical postprocessing of the outcomes from very few random measurements specifically designed in Ref. \cite{Huang2020}. It leads to the data processing in this protocol is cheap - both in memory and runtime - and can be massively parallelized \cite{Elben2020}.
In addition, it is known from
\cite{Mac1995,Duan2022} that the spectrum of $\rho_{AB}^{\Gamma}$
can be completely determined with all known PT-moments. Then the
PT-moments provide an operational method of detecting entanglement
with the help of PPT criterion. The measurability of PT-moments
bridges the practical limitations of the PPT criterion, and allows
it to be used for experimental entanglement detection. For this
advantage, the PT-moment-based entanglement detection has aroused
great interest and been widely investigated recently
\cite{Liu2022,TGZhang2022,Wang2022}.

Moreover, the PT-moments also provide a characterization of states.
It is natural to ask whether there is a separable or entangled state
compatible with the given data of PT-moments. Due to the difficulty
of measuring all the PT-moments, the authors in \cite{Yu2021}
characterized the separable states supported on $\complex^d$ using
the first three PT-moments, where the first PT-moment is normalized
as one. Specifically, it is known from \cite{Yu2021} that there is a
separable state compatible with the given pair of PT-moments
$(p_2,p_3)$ if and only if $p_3$ is lower and upper bounded by two
functions of variable $p_2\in[\frac 1d,1]$. In \cite{Zhang2022} we
considered the similar question corresponding to two-qubit entangled
states, i.e. whether there is an entangled two-qubit state
compatible with the given pair $(p_2,p_3)$. We answered the above
question by determining the accurate lower bound function for $p_3$
of variable $p_2\in[\frac 13, 1]$ \cite{Zhang2022}. Thus, combining
with the results corresponding to two-qubit separable states, we
have already known the 2D feasible region of $(p_2,p_3)$
corresponding to all two-qubit states and the dividing curve between
the two regions corresponding to separable and entangled states
respectively.

The 2D region of $(p_2,p_3)$ for all two-qubit states is lack of the
information of $p_4$, and thus it is incomplete and just a
projection of the whole 3D region of $(p_2,p_3,p_4)$. In order to
describe the complete graph of all two-qubit states, we further
study the relation between $p_4$ and the pair $(p_2,p_3)$ in this
paper. The complete graph of all two-qubit states is a 3D region
in the $(p_2,p_3,p_4)$-coordinate system, where the pair $(p_2,p_3)$
is restricted in the 2D feasible region above-mentioned
\cite{Zhang2022}. Then the characterization of the complete graph
of all two-qubit states can be mathematically formulated as the
problem of bounding $p_4$ with binary functions in terms of two
variables $p_2$ and $p_3$, where the pair $(p_2,p_3)$ is restricted
in the feasible region. By virtue of the linear relation between
$p_4$ and the determinant of $\rho_{AB}^{\Gamma}$ for the fixed pair
$(p_2,p_3)$, we equivalently transform the original problem
to an optimization problem with respect to a variable $s$ which is a
sum of the largest three eigenvalues of $\rho_{AB}^{\Gamma}$. By
figuring out this optimization problem we finally determine the
lower bound function $F^-(p_2,p_3)$ and upper bound function
$F^+(p_2,p_3)$ in Theorem \ref{thm:main}, such that
$F^-(p_2,p_3)\leqslant p_4\leqslant F^+(p_2,p_3)$ for any fixed pair
$(p_2,p_3)$ in the feasible region. Furthermore, in Proposition
\ref{prop:ent} we identify the border surface which separates the
two regions corresponding to entangled and separable states
respectively. Due to the measurability of PT-moments, we obtain an
operational criterion for the detection of two-qubit entanglement.
In order to make the process of deriving $F^+(p_2,p_3)$ and
$F^-(p_2,p_3)$ clearer, we summarize the calculating procedure of
the maximum and minimum values of $p_4$ for a given pair $(p_2,p_3)$
in the feasible region. As examples we plot the graph of the
family of Werner states and the family of Bell-diagonal states in
Figure \ref{fig:werner3d} and Figure \ref{fig:bell3d} respectively,
by following the above calculating procedure. Generally, we further visualize
the 3D region corresponding to the set of all two-qubit states in
Figure \ref{fig:2qubitsepvsent}, and mark the separable region and
entangled region in different colors. Finally, we compare the criterion in terms of the triple $(p_2,p_3,p_4)$ with two common entanglement measures, i.e. the negativity and the concurrence, and consequently conclude that the criterion proposed in this paper has the same performance as such two entanglement measures in the aspect of detecting two-qubit entanglement.

The paper is organized as follows. In Sect.~\ref{sec:pre}, we
mathematically formulate the expression of each PT-moment, recall
the feasible region of the pair $(p_2,p_3)$ specifically, and
clarify some essential notations frequently used to derive main
results. In Sect.~\ref{sec:mre}, we specifically explain the process
of deriving the lower and upper bound functions for $p_4$ step by
step. Some derivation details are put in appendices for simplicity. In Sect. \ref{sec:exam}, we summarize the calculating
procedure of the maximum and minimum values of $p_4$ when inputting
a given pair $(p_2,p_3)$. We specifically study two families of
states by following the calculating procedure in this section. In
Sect.~\ref{sec:fig}, we present the visualization of
all two-qubit states via PT-moments. In Sect. \ref{sec:com}, we compare our proposed criterion via PT-moments with two common entanglement measures on the entanglement detection.
Finally, we conclude in Sect.~\ref{sec:con}.

\section{Preliminary}\label{sec:pre}

For any bipartite state $\rho_{AB}$ supported on the Hilbert space
${\cal H}_{AB}\cong{\mathbb C}^{d_A}\otimes\mathbb C^{d_B}$, its
PT-moments \cite{Yu2021} are defined in terms of the spectrum of its
partial transpose $\rho_{AB}^\Gamma$ with respect to one of both
subsystems as:
\begin{eqnarray}\label{eq:def}
p_k=\Tr{\Br{\rho^\Gamma_{AB}}^k} =\sum_{j=1}^{d_Ad_B}
x_j^k\quad(1\leqslant k\leqslant d_Ad_B)
\end{eqnarray}
where $x_j$'s are the eigenvalues of $\rho_{AB}^\Gamma$ sorted in
the descending order. We call $\bsp^{(d)}=(p_1,\cdots,p_d)$ the
PT-moment vector where $d=d_Ad_B$, and always assume $p_1=1$ for
convenience. It follows from \cite{Mac1995} that the spectrum of
$\rho_{AB}^{\Gamma}$ can be completely determined with all known
PT-moments, for which whether $\rho_{AB}$ is an NPT state can be
verified.

In this paper, we focus on the two-qubit system. According to the
above description, there are four PT-moments in total for any
two-qubit state, i.e. $p_1,p_2,p_3,p_4$ defined by Eq.
\eqref{eq:def}. For two-qubit states, we have already identified the
2D region comprised of a pair of PT-moments namely $(p_2,p_3)$ in
our recent paper \cite{Zhang2022}. In order to obtain a more
comprehensive picture of all two-qubit states by PT-moments, we
further introduce the fourth PT-moment $p_4$ to characterize the
complete graph of all two-qubit states. Specifically, we shall identify
the 3D regions corresponding to the set of entangled states and that
of separable states respectively in terms of a triple of PT-moments
namely $(p_2,p_3,p_4)$. First, we recall the relation between $p_2$
and $p_3$ derived in \cite{Zhang2022}, where we bounded $p_3$ with
functions of variable $p_2$ as follows. It is known from
\cite{Zhang2022} that
\begin{eqnarray}
\label{eq:p2-p3}
\frac14\leqslant p_2\leqslant1,\quad f^-(p_2)\leqslant p_3\leqslant
f^+(p_2),
\end{eqnarray}
where
\begin{eqnarray}
\label{eq:p2-p3-f}
f^\pm(p_2)=\frac{3(6p_2-1)\pm\sqrt{3}(4p_2-1)^\frac32}{24}.
\end{eqnarray}
Denote by
\begin{eqnarray}\label{eq:A}
\cA:=\Set{(p_2,p_3): \frac14\leqslant
p_2\leqslant1,~~~f^-(p_2)\leqslant p_3\leqslant f^+(p_2)}
\end{eqnarray}
the feasible region of the pair of PT-moments $(p_2,p_3)$ for all two-qubit states. Furthermore, from \cite[Fig. 1]{Zhang2022} there is a curve represented by $p_3=\phi_4(p_2)$ which divides the separable part (corresponding to separable states) and entangled part (corresponding to entangled states) of 2D region $\cA$. The function $p_3=\phi_4(p_2)$ has been specifically formulated \cite[Eq. (10)]{Zhang2022} as
\begin{equation}
  \phi_4(p_2)=
\begin{cases}
  \frac{3p_2-1}{2}, & p_2\in[\frac 12, 1] \\
  \frac{2(9p_2-2)-\sqrt{2}(3p_2-1)^{\frac 32}}{18}, & p_2\in[\frac 13, \frac 12] \\
  \frac{3(6p_2-1)-\sqrt{3}(4p_2-1)^{\frac 32}}{24}, & p_2\in[\frac 14, \frac 13]
\end{cases}.
\end{equation}
One can verify that $\phi_4(p_2)\geqslant p_2^2, ~\forall p_2\in [\frac 14, 1]$, and the equality holds only when $p_2=\frac 14, \frac 13, 1$. Hence, there are infinite pairs $(p_2,p_3)\in\cA$ satisfying $\phi_4(p_2)>p_3>p^2_2$ for $p_2\in(\frac13,1)$. Recall that if $\rho_{AB}$ is PPT, then its pair of PT-moments must satisfy $p_3\geqslant p_2^2$ \cite{Elben2020}. Comparing this criterion with the criterion we previously proposed in Ref. \cite{Zhang2022}, there are entangled states which can be detected by $\phi_4(p_2)>p_3$ but cannot be detected by $p_3>p_2^2$. It shows our previous criterion is more powerful than that in Ref. \cite{Elben2020}.
Analogous to the 2D region given by $\cA$, we continue to derive the 3D region comprised of the triple of PT-moments namely $(p_2,p_3,p_4)$ for all two-qubit states. Because the last PT-moment $p_4$ for two-qubit states is involved, the 3D region and the entanglement detection criterion in terms of $(p_2,p_3,p_4)$ should be more complete. For convenience, we introduce some essential notations which will be frequently used in the expression of our main results.

The focused problem is to bound $p_4$ with binary functions
in $p_2$ and $p_3$. According to Eq.~\eqref{eq:def}, each PT-moment
of a state is closely related to the spectrum of its partial
transpose. Denote by $\rho_{AB}$ an arbitrary two-qubit state, and
$\rho_{AB}^{\Gamma}$ is the partial transpose. There are some
constraints on the spectrum of $\rho_{AB}^{\Gamma}$. It
follows from Ref. \cite{Rana2013pra} that $\rho_{AB}$ is NPT if and only
if $\rho_{AB}^{\Gamma}$ has exact one negative value, and the
eigenvalues of $\rho_{AB}^{\Gamma}$ all belong to the interval
$[-\frac12, 1]$. Suppose $x\geqslant y\geqslant z$ are the three
nonnegative eigenvalues of $\rho_{AB}^\Gamma$, and thus the fourth
eigenvalue is $1-x-y-z$. Then for any fixed pair of PT-moments
$(p_2,p_3)\in\cA$, all of three parameters must satisfy the
following constraints:
\begin{eqnarray}
\label{eq:maincons-1}
&&1\geqslant x\geqslant y\geqslant z\geqslant 0,\\
\label{eq:maincons-2}
&&z\geqslant 1-x-y-z \geqslant -\frac12,\\
\label{eq:maincons-3}
&&x^2+y^2+z^2+(1-x-y-z)^2=p_2,\\
\label{eq:maincons-4} &&x^3+y^3+z^3+(1-x-y-z)^3=p_3.
\end{eqnarray}
Define the objective function as $s\equiv s(x,y,z):=x+y+z$. For any
fixed pair of PT-moments $(p_2,p_3)\in\cA$, denote by $\cV(p_2,p_3)$
the set of feasible triples $(x,y,z)$ meeting all the above
constraints. As we shall see in Sect.~\ref{sec:mre}, the following
two values are essential to derive our main results. For any pair
$(p_2,p_3)\in\cA$, define
\begin{eqnarray}
\label{eq:smax}
s_{\max}&\equiv& s_{\max}(p_2,p_3):=\max_{\cV(p_2,p_3)} s(x,y,z),\\
\label{eq:smin} s_{\min}&\equiv&
s_{\min}(p_2,p_3):=\min_{\cV(p_2,p_3)} s(x,y,z).
\end{eqnarray}
Without ambiguity, we may use $s_{\max}$ and $s_{\min}$ to denote
$s_{\max}(p_2,p_3)$ and $s_{\min}(p_2,p_3)$, respectively.

Although for any given pair $(p_2,p_3)$, the two values $s_{\max}$
and $s_{\min}$ can be numerically calculated, the analytical
expressions of the two functions $s_{\max}(p_2,p_3)$ and
$s_{\min}(p_2,p_3)$ given by Eqs. \eqref{eq:smax} and
\eqref{eq:smin} are quite complicated to derive. Here we made two
rough conclusions on $s_{\max}$ and $s_{\min}$. First, combining
Eqs. \eqref{eq:maincons-1} and \eqref{eq:maincons-2} leads to
$1\geqslant x\geqslant y\geqslant z\geqslant 1-s\geqslant-\frac12$.
It implies that the value $s=x+y+z$ satisfies $s\geqslant 3(1-s)$
and $1-s\geqslant-\frac12$. It is not hard to figure out
$s\in[\tfrac34,\frac32]$, and thus
$[s_{\min},s_{\max}]\subseteq[\tfrac34,\frac32]$ from Eqs.
\eqref{eq:smax} and \eqref{eq:smin}. Second, we claim that the
maximum and minimum values $s_{\max}$ and $s_{\min}$ are reached
only if all of three variables $x,y,z$ are equal. We present the
proof of this claim in Appendix \ref{sec:appsmaxmin}.

\section{Characterization of the 3D region in terms of PT-moments}
\label{sec:mre}

As we have obtained the 2D region comprised of the pair of
PT-moments $(p_2,p_3)$ \cite{Zhang2022}, in order to obtain the
whole 3D region comprised of the triple of PT-moments
$(p_2,p_3,p_4)$, we only need to bound the fourth PT-moments $p_4$
with binary functions of two variables $p_2$ and $p_3$. In this
section we derive the lower bound function $F^-(p_2,p_3)$ and upper
bound function $F^+(p_2,p_3)$ such that $F^-(p_2,p_3) \leqslant
p_4\leqslant F^+(p_2,p_3)$ for all two-qubit states. Such two
functions can be derived through the following steps.

First of all, it follows from Ref. \cite{Mac1995} that the determinant of
$\rho_{AB}^\Gamma$ has an analytical expression in terms of
$p_2,p_3,p_4$, i.e.
\begin{eqnarray}
\label{eq:det}
\det(\rho^\Gamma_{AB}) =\frac{3p^2_2 - 6p_2 + 8p_3 - 6 p_4 + 1}{4!}.
\end{eqnarray}
From the above equation we calculate $p_4$ as
\begin{eqnarray}
\label{eq:det-p4}
p_4=\frac{3 p_2^2-6 p_2+8 p_3+1}6 -4\det(\rho^\Gamma_{AB}).
\end{eqnarray}
Let $F(p_2,p_3):=\frac{3 p_2^2-6 p_2+8 p_3+1}6$. Then $p_4$ also
reads $p_4=F(p_2,p_3)-4\det(\rho^\Gamma_{AB})$ for simplicity. Since
$(p_2,p_3)\in\cA$ are fixed arbitrarily, it follows from Eq.
\eqref{eq:det-p4} that the maximum of $p_4$ is achieved if and only
if the minimum of $\det(\rho^\Gamma_{AB})$ is achieved, and
analogously the minimum of $p_4$ is achieved if and only if the
maximum of $\det(\rho^\Gamma_{AB})$ is achieved. That is, for any
fixed pair $(p_2,p_3)\in\cA$,
\begin{eqnarray}
\label{eq:F+todet-}
F^+(p_2,p_3) := \frac{3 p_2^2-6 p_2+8 p_3+1}6 -
4m(p_2,p_3)=F(p_2,p_3)-4m(p_2,p_3),\\
\label{eq:F-todet+}
F^-(p_2,p_3) := \frac{3 p_2^2-6 p_2+8 p_3+1}6 -
4M(p_2,p_3)=F(p_2,p_3)-4M(p_2,p_3),
\end{eqnarray}
where $m(p_2,p_3)$ and $M(p_2,p_3)$ are respectively the minimum and
maximum of $\det(\rho^\Gamma_{AB})$ over those two-qubit states
whose second and third PT-moments are respectively equal to the
given $p_2$ and $p_3$. Thus, in this first step, we equivalently
transform the problem of deriving $F^+(p_2,p_3)$ and $F^-(p_2,p_3)$
to the problem of bounding the determinant $\det(\rho^\Gamma_{AB})$.

The second step is to accurately bound
$\det(\rho^\Gamma_{AB})$, i.e. the lower bound $m(p_2,p_3)$ and upper bound $M(p_2,p_3)$ for any
given pair $(p_2,p_3)\in\cA$. Recall that we have assumed $x,y,z$ and $1-x-y-z$ are
the four eigenvalues of $\rho_{AB}^\Gamma$ in the descending order.
It implies $\det(\rho_{AB}^\Gamma)=xyz(1-x-y-z)$. Recall that
for any fixed pair $(p_2,p_3)\in\cA$, the triple of variables
$(x,y,z)$ satisfies the constraints given by Eqs.
\eqref{eq:maincons-1} --- \eqref{eq:maincons-4}. We have already
known that $\rho_{AB}$ is an NPT state if and only if
$\det(\rho^\Gamma_{AB})<0$ \cite{Rana2013pra,Shen2020}. It implies
that $\det(\rho^\Gamma_{AB})$ attains its maximal value on some PPT
states, i.e. separable states, and attains its minimal value on some
NPT states, i.e., entangled states. Hence, we first propose the
global maximal and minimal values of
$\det(\rho^\Gamma_{AB})$ over the set of all two-qubit states,
namely $\density{\complex^2\ot\complex^2}$.
For any fixed pair $(p_2,p_3)\in \cA$, we conclude that
\begin{eqnarray}
\label{eq:mMbound}
-\frac1{16}\leqslant m(p_2,p_3)\leqslant
\det(\rho^\Gamma_{AB})\leqslant M(p_2,p_3)\leqslant
\frac1{256}.
\end{eqnarray}
The specific calculation of the global maximal and minimal values of $\det(\rho_{AB}^\Gamma)$ is put in Appendix \ref{sec:gdet}.
Combining the above observation with Eqs. \eqref{eq:F+todet-} ---
\eqref{eq:F-todet+}, we also estimate the difference between
$F^+(p_2,p_3)$ and $F^-(p_2,p_3)$ as
\begin{eqnarray}
0\leqslant F^+(p_2,p_3) - F^-(p_2,p_3) = 4[M(p_2,p_3)-m(p_2,p_3)]
\leqslant \frac{17}{64},\quad \forall(p_2,p_3)\in\cA.
\end{eqnarray}

Next, we shall determine $m(p_2,p_3)$ and $M(p_2,p_3)$
accurately for any pair $(p_2,p_3)\in\cA$. We transform this problem to the optimization problem on the objective function $P(s|p_2,p_3)$ over the closed interval $[s_{\min},s_{\max}]$, where $s_{\max}$ and $s_{\min}$ have been clarified by Eqs. \eqref{eq:smax} and \eqref{eq:smin} respectively. That is,
\begin{eqnarray}
\label{eq:MtoP}
M(p_2,p_3)&=&\max_{s\in[s_{\min},s_{\max}]}\Phi(s|p_2,p_3) =
-\min_{s\in[s_{\min},s_{\max}]}P(s|p_2,p_3),\\
\label{eq:mtoP}
m(p_2,p_3)&=&\min_{s\in[s_{\min},s_{\max}]}\Phi(s|p_2,p_3) =
-\max_{s\in[s_{\min},s_{\max}]}P(s|p_2,p_3),
\end{eqnarray}
where
\begin{equation}
\label{eq:-phi-p2-p3}
\begin{aligned}
P(s|p_2,p_3)&:=-\Phi(s|p_2,p_3) \\
&=s^4-3s^3+\frac{7-p_2}2s^2+\frac{3p_2+2p_3-11}6s+\frac{1-p_3}3,
\end{aligned}
\end{equation}
for $s\in[s_{\min},s_{\max}]$. We put the detailed transformation process in Appendix \ref{sec:MmtoP}.

The third step is to specifically optimize $P(s|p_2,p_3)$ over
$[s_{\min},s_{\max}]$. We consider $P(s|p_2,p_3)$ as a parametric polynomial with respect to variable $s$, and employ the typical method to figure out this optimization problem. We put the specific derivation process in Appendix \ref{sec:optP}.
Based on the analytical discussion in Appendix \ref{sec:optP} we
can draw the conclusion as follows.
\begin{prop}\label{prop:Ploex}
Let
$P(s|p_2,p_3):=s^4-3s^3+\frac{7-p_2}2s^2+\frac{3p_2+2p_3-11}6s+\frac{1-p_3}3$
be a function of variable $s$, where the pair of parameters
$(p_2,p_3)$ are fixed in $\cA$. Let $r_1\geqslant r_2\geqslant r_3$
be the three roots formulated in Eq. \eqref{eq:Delta<0roots-3}, such
that the first derivative $\frac{\dif}{\dif s}P(s|p_2,p_3)$ is zero.
Then,
\begin{eqnarray}
\min_{s\in[s_{\min},s_{\max}]}P(s|p_2,p_3)&=&\min\Set{v_1,v_2,v_3,v_4},\label{eq:Plomin}\\
\max_{s\in[s_{\min},s_{\max}]}P(s|p_2,p_3)&=&\max\Set{v_1,v_2,v_3,v_4}.\label{eq:Plomax}
\end{eqnarray}
where
\begin{eqnarray}
\notag
v_1&:=&P\Pa{s_{\max}|p_2,p_3},\quad v_2:=P(\min\set{r_1,s_{\max}}|p_2,p_3),\\
\notag
v_3&:=&P\Pa{\max\set{s_{\min},r_2}|p_2,p_3},\quad
v_4:=P(s_{\min}|p_2,p_3).
\end{eqnarray}
\end{prop}
Proposition \ref{prop:Ploex} tells us that the maximal and minimal values of $P(s|p_2,p_3)$ over $[s_{\min},s_{\max}]$ are both included in a same set with only four values, which is an essential simplification.

The final step is to identify the two boundary functions
$F^-(p_2,p_3)$ and $F^+(p_2,p_3)$ for the fourth PT-moment $p_4$.
Due to the equivalent relations built in the first and second steps,
we may formulate $F^-(p_2,p_3)$ and $F^+(p_2,p_3)$ with the local
extreme values of $P(s|p_2,p_3)$ for $s\in[s_{\min},s_{\max}]$. The
results can be summarized in the following theorem.
\begin{thrm}
\label{thm:main} There is a two-qubit state $\rho_{AB}$ compatible
with the given PT-moment vector $\bsp^{(4)}=(p_1,p_2,p_3,p_4)$,
where $p_1=1$, if and only if $F^-(p_2,p_3) \leqslant p_4\leqslant
F^+(p_2,p_3)$, where $(p_2,p_3)\in\cA$, that is, $p_2,p_3,p_4$
satisfy the following conditions:
\begin{eqnarray}
\label{eq:thm-p2}
&&\frac14\leqslant p_2\leqslant1,\\
\label{eq:thm-p3}
&&f^-(p_2)\leqslant p_3\leqslant f^+(p_2), \\
\label{eq:thm-p4}
&&F^-(p_2,p_3) \leqslant p_4\leqslant F^+(p_2,p_3),
\end{eqnarray}
where $f^{\pm}(p_2)$ are given by Eq. \eqref{eq:p2-p3-f}, and
$F^\pm(p_2,p_3)$ are determined by
\begin{eqnarray}
F^-(p_2,p_3) &=& F(p_2,p_3)+4\min_{s\in[s_{\min},s_{\max}]}P(s|p_2,p_3), \label{eq:upper}\\
F^+(p_2,p_3)&=&
F(p_2,p_3)+4\max_{s\in[s_{\min},s_{\max}]}P(s|p_2,p_3),
\label{eq:lower}
\end{eqnarray}
via
\begin{eqnarray}
\label{eq:mid}
F(p_2,p_3) := \frac{3 p_2^2-6 p_2+8 p_3+1}6,
\end{eqnarray}
and the two extreme values given by Eqs. \eqref{eq:Plomin} and
\eqref{eq:Plomax}, respectively.
\end{thrm}

\begin{proof}
It follows from Ref. \cite{Zhang2022} that the feasible set of
$(p_2,p_3)$ for all two-qubit states is $\cA$ given by Eq.
\eqref{eq:A}. It implies that Eqs. \eqref{eq:thm-p2} and
\eqref{eq:thm-p3} are necessarily satisfied. Combining Eqs.
\eqref{eq:F-todet+} and \eqref{eq:MtoP} we obtain the expression of
$F^-(p_2,p_3)$, i.e. Eq. \eqref{eq:upper}. Similarly, combining Eqs.
\eqref{eq:F+todet-} and \eqref{eq:mtoP} we obtain the expression of
$F^+(p_2,p_3)$, i.e. Eq. \eqref{eq:lower}. Furthermore, the two
extreme values $\min_{s\in[s_{\min},s_{\max}]}P(s|p_2,p_3)$ and
$\max_{s\in[s_{\min},s_{\max}]}P(s|p_2,p_3)$ has been derived in
Proposition \ref{prop:Ploex}. This completes the proof.
\end{proof}

\begin{remark}
Let
\begin{eqnarray}
Q(s|p_2,p_3)&:=&F(p_2,p_3)+4P(s|p_2,p_3) \notag\\
&=&4\Pa{
s^4-3s^3+\frac{7-p_2}2s^2+\frac{3p_2+2p_3-11}6s+\frac{p_2^2-2
p_2+3}8}.
\end{eqnarray}
Then we see that
\begin{eqnarray}
F^{+/-}(p_2,p_3) =
{\max/\min}_{s\in[s_{\min},s_{\max}]}Q(s|p_2,p_3).
\end{eqnarray}
\end{remark}
The conditions given by Eqs. \eqref{eq:thm-p2} - \eqref{eq:thm-p4}
determine a 3D region in the $(p_2,p_3,p_4)$-coordinate system,
which is corresponding to the set of all two-qubit states. Thus, we
characterize the 3D region of all two-qubit states via complete
PT-moments by Theorem \ref{thm:main}.

Based on the above result, we can give an operational criterion for
the detection of entanglement existing in two-qubit states as
follows.
\begin{prop}
\label{prop:ent}
\begin{enumerate}[(i)]
\item The entangled two-qubit state $\rho_{AB}$ is compatible with a PT-moment vector $\bsp^{(4)}$ if
and only if
\begin{eqnarray}
\label{eq:entc}
F(p_2,p_3)<p_4\leqslant F^+(p_2,p_3),\quad \forall(p_2,p_3)\in\cA.
\end{eqnarray}
\item The separable two-qubit state $\rho_{AB}$ is compatible with a PT-moment vector $\bsp^{(4)}$ if
and only if
\begin{eqnarray}
\label{eq:sepc}
F^-(p_2,p_3) \leqslant p_4\leqslant F(p_2,p_3),\quad
\forall(p_2,p_3)\in\cA.
\end{eqnarray}
\end{enumerate}
\end{prop}

\begin{proof}
According to Theorem \ref{thm:main}, it is equivalent to proving any two-qubit state $\rho_{AB}$ is entangled if and only if the following inequality holds:
\begin{eqnarray}
\label{eq:entid}
p_4>F(p_2,p_3),
\end{eqnarray}
where $F(p_2,p_3):=\frac{3p^2_2-6p_2+8p_3+1}6$. This follows directly from the fact that
\begin{eqnarray}
\det(\rho^\Gamma_{AB}) =\frac{3p^2_2 - 6p_2 + 8p_3 - 6 p_4 + 1}{4!}.
\end{eqnarray}
It follows from \cite{Rana2013pra,Shen2020} that $\rho_{AB}$ is
entangled if and only if $\det(\rho^\Gamma_{AB})<0$. It follows that
$3p^2_2 - 6p_2 + 8p_3 - 6 p_4 + 1<0$, and equivalently Eq.
\eqref{eq:entid} holds. This completes the proof.
\end{proof}
According to Proposition \ref{prop:ent}, we identify the dividing surface represented by $p_4=F(p_2,p_3)$ which separates the whole 3D region of all two-qubit states into two parts corresponding to separable states and entangled states respectively.

\section{Calculating procedure with examples}\label{sec:exam}

According to the above discussion, we can list the calculating
procedure of $F^{+/-}(p_2,p_3)$ (corresponding to $p^{\max/\min}_4$)
for $(p_2,p_3)\in\cA$.
\begin{center}
\begin{algorithm}[H]\caption{The max/min value of $p_4$}
 \KwIn{$(p_2,p_3)\in\cA$} \KwOut{$p^{\max/\min}_4$}
 choose randomly $(p_2,p_3)\in\cA$\;
 calculate $s_{\max/\min}\equiv s_{\max/\min}(p_2,p_3)$ via Eqs.~\eqref{eq:smax}/\eqref{eq:smin}\;
 calculate $M(p_2,p_3)/m(p_2,p_3)$ via
 Eqs.~\eqref{eq:MtoP}/\eqref{eq:mtoP}\;
 calculate $F^{\pm}(p_2,p_3)$ via
 Eqs.~\eqref{eq:F+todet-}/\eqref{eq:F-todet+}\;
 Output $p^{\max}_4=F^+(p_2,p_3)$ and $p^{\min}_4=F^-(p_2,p_3)$.
\end{algorithm}
\end{center}
Following the above calculating procedure, we list some analytical examples in Table~\ref{tab} below.
\begin{table}[h]
\centering
\begin{tabular}{|c|c|c|c|c|c|c|}
\hline $(p_2,p_3)\in\cA$ & $s_{\min}$ & $s_{\max}$ & $m$ & $M$ & $F^-$ & $F^+$\\
\hline $(\tfrac14,\tfrac1{16})$ & $\tfrac34$ & $\tfrac34$ & $\tfrac1{256}$ & $\tfrac1{256}$ & $\tfrac1{64}$ & $\tfrac1{64}$\\
\hline $(\tfrac13,\tfrac5{36})$ & $\tfrac56$ & $\tfrac56$ & $\frac1{432}$ & $\frac1{432}$ & $\frac7{108}$ & $\frac7{108}$\\
\hline $(\tfrac13,\tfrac19)$ & $1$ & $1$ & $0$ & $0$ & $\frac1{27}$ & $\frac1{27}$\\
\hline $(\tfrac12,\tfrac{6+\sqrt{3}}{24})$ & $\frac{9+\sqrt{3}}{12}$ & $\frac{9+\sqrt{3}}{12}$ & $\frac{2\sqrt{3}-3}{576}$ & $\frac{2\sqrt{3}-3}{576}$ & $\frac{7+2 \sqrt{3}}{48}$ & $\frac{7+2 \sqrt{3}}{48}$\\
\hline $(\tfrac12,\tfrac14)$ & $1$ & $\frac{3+\sqrt{2}}4$ & $-\frac1{256}$ & $0$ & $\frac18$ & $\frac9{64}$\\
\hline $(\tfrac12,\tfrac{6-\sqrt{3}}{24})$ & $\frac{3+\sqrt{3}}4$ & $\frac{3+\sqrt{3}}4$ &$-\frac{2\sqrt{3}+3}{576}$ & $-\frac{2\sqrt{3}+3}{576}$ & $\frac{7-2 \sqrt{3}}{48}$ & $\frac{7-2 \sqrt{3}}{48}$\\
\hline $(1,1)$ & $1$ & $1$ & $0$ & $0$ & $1$ & $1$\\
\hline $(1,\tfrac14)$ & $\frac32$ & $\frac32$ & $-\frac1{16}$ & $-\frac1{16}$ & $\frac14$ & $\frac14$\\
\hline
\end{tabular}
\caption{Specific examples with analytical computations}\label{tab}
\end{table}

In the following we shall discuss two families of states with extensive applications in quantum information processing by virtue of the above calculating procedure.

\begin{exam}[The family of Werner states]
In 1989, Werner analytically constructed a family of $\bsU\ot\bsU$
invariant states to investigate local hidden variable (LHV) models.
As a toy model, we consider the two-qubit Werner states formulated
as
\begin{eqnarray}
\label{eq:defwerners} \rho_W(w) =
w\out{\psi^-}{\psi^-}+(1-w)\frac{\mathbf{1}_4}4
\end{eqnarray}
where $\ket{\psi^-}=\frac{\ket{01}-\ket{10}}{\sqrt{2}}$ and
$w\in[0,1]$. By calculation, the 2nd, 3rd, and 4th PT-moments of
two-qubit Werner state $\rho_W(w)$ are, respectively,
\begin{eqnarray*}
\begin{cases}
p_2=\frac{1+3w^2}4\\
p_3=\frac{-6 w^3+9 w^2+1}{16}\\
p_4=\frac{21 w^4-24 w^3+18 w^2+1}{64}
\end{cases}\quad (w\in[0,1])
\end{eqnarray*}
The region for the family of Werner states in the $(p_2,p_3,p_4)$-coordinate system is just a curve, which is visualized in the following
Figure~\ref{fig:werner3d}.
\begin{figure}[h!]\centering
{\begin{minipage}[b]{0.58\linewidth}
\includegraphics[width=1\textwidth]{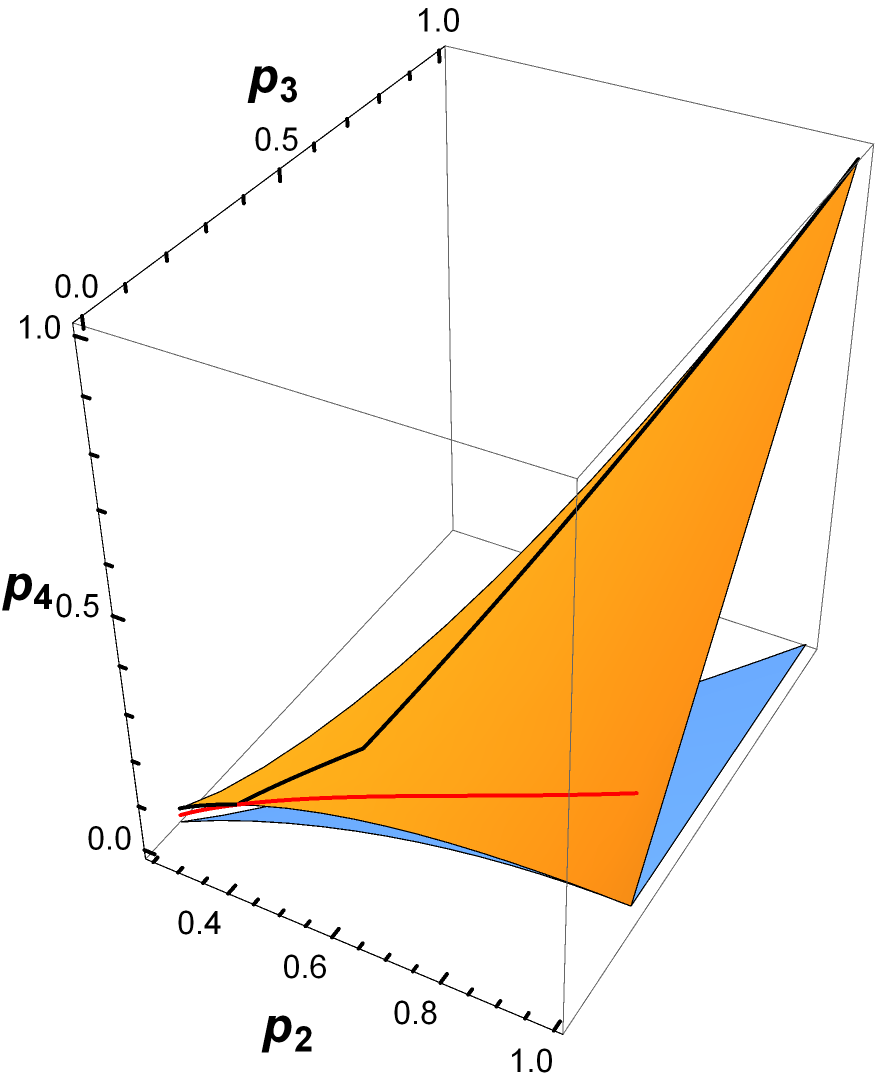}
\end{minipage}}
\caption{(Color Online) The family of Werner states corresponds to the red curve in the $(p_2,p_3,p_4)$-coordinate system,
crossing the yellow surface represented by
$p_4=F(p_2,p_3)$ at the intersection point $P_0$ with
$(p_2,p_3,p_4)=(\frac13,\frac19,\frac1{27})$.
It is also known from Ref. \cite{Zhang2022} that in the $(p_2,p_3)$-coordinate system the family of Werner states corresponds to the lower boundary curve of 2D region $\cA$ marked in blue on the bottom plane, where $p_2\in[\frac 14, 1]$. It means that the lower boundary of the blue region is the projection of the red curve.
Since the yellow surface $p_4=F(p_2,p_3)$ divides the entangled and separable regions, the part of the red curve over the yellow surface (i.e. those points with $p_2>\frac 13$) are
entangled Werner states, and the part of the red curve below the yellow
surface (i.e. those points with $p_2\leqslant\frac 13$) are separable Werner states.
Accordingly for the projection of the red curve, the part of the lower boundary curve of the blue region with $\frac 14 \leqslant p_2\leqslant \frac 13$ are separable Werner states, and that part with $\frac 13 <p_2\leqslant 1$ are entangled Werner states \cite{Zhang2022}.
In addition, the black curve on the yellow surface is projected as the curve given by $p_3=\phi_4(p_2)$ which divides the separable part and entangled part of 2D region $\cA$ \cite{Zhang2022}. Moreover, the projection of the black curve with $\frac 14 \leqslant p_2\leqslant \frac 13$ coincides with the corresponding part of the lower boundary curve of the blue region.}\label{fig:werner3d}
\end{figure}
\end{exam}

\begin{exam}[The family of Bell-diagonal states]
The Bell-diagonal states \cite{Horodecki1996} in the two-qubit
system can be written as
\begin{eqnarray}
\label{eq:defbelldiag}
\rho_{\text{Bell}}=\frac14(\mathbf{1}_4+\sum_{i=1}^3 t_i
\sigma_i\otimes\sigma_i),
\end{eqnarray}
where $\sigma_i$ for $i=1,2,3$ are the three Pauli operators as
$\sigma_1=|0\rangle\langle 1|+ |1\rangle\langle 0|$, $\sigma_2={\rm
i}|0\rangle\langle 1|- {\rm i}|1\rangle\langle 0|$,
$\sigma_3=|0\rangle\langle 0|- |1\rangle\langle 1|$. Hence, a
Bell-diagonal state is specified by three real variables $t_1, t_2$,
and $t_3$ such that the formula in Eq. \eqref{eq:defbelldiag} represents a normalized positive semidefinite operator, which is equivalent to the following constraints:
\begin{eqnarray}
\label{eq:bdiag-1}
\begin{cases}
1-t_1-t_2-t_3\geqslant 0,\\
1-t_1+t_2+t_3\geqslant 0,\\
1+t_1-t_2+t_3\geqslant 0,\\
1+t_1+t_2-t_3\geqslant 0.
\end{cases}
\end{eqnarray}
Denote by $D_{\text{Bell}}$ the set of tuples $(t_1,t_2,t_3)$
satisfying the above system of inequalities. Because all the four
eigenvalues of a two-qubit state are in $[0,1]$, it follows from Eq.
\eqref{eq:defbelldiag} that $t_i\in[-1,1]$ for $i=1,2,3$. That is,
$D_{\text{Bell}}\subset[-1,1]^3$. Furthermore, the Bell-diagonal
states can be graphically described by a tetrahedron. One can show
that a Bell-diagonal state is separable if and only if
$\abs{t_1}+\abs{t_2}+\abs{t_3}\leqslant 1$ holds. Graphically, the
set of Bell-diagonal states is a tetrahedron and the set of
separable Bell-diagonal states is an octahedron
\cite{Horodecki1996}, which is denoted by $D_{\text{Bellsep}}$. By
calculation, all eigenvalues of $\rho^{\Gamma}_{\text{Bell}}$ are
given by $\frac{1+t_1-t_2-t_3}4, \frac{1-t_1+t_2-t_3}4,
\frac{1+t_1+t_2+t_3}4,\frac{1-t_1-t_2+t_3}4$. It follows from the definition that
\begin{eqnarray}
\label{eq:bdiag-2}
\begin{cases}
p_2=\frac{1+\sum^3_{i=1}t^2_i}4\\
p_3=\frac{1+6t_1t_2t_3+3\sum^3_{i=1}t^2_i}{16}\\
p_4=\frac{1+24t_1t_2t_3+6\sum^3_{i=1}t^2_i+\sum^3_{i=1}t^4_i+6\sum_{1\leqslant
i<j\leqslant3}t^2_i t^2_j}{64}
\end{cases}
\end{eqnarray}
By virtue of the parametric equations formulated in Eq. \eqref{eq:bdiag-2}, we can draw the 3D region corresponding to the family of Bell-diagonal states in the
$(p_2,p_3,p_4)$-coordinate system , as depicted in Figure~\ref{fig:bell3d} below.
\begin{figure}[h!]
\subfigure[]{\begin{minipage}[b]{0.45\linewidth}
\includegraphics[width=1\textwidth]{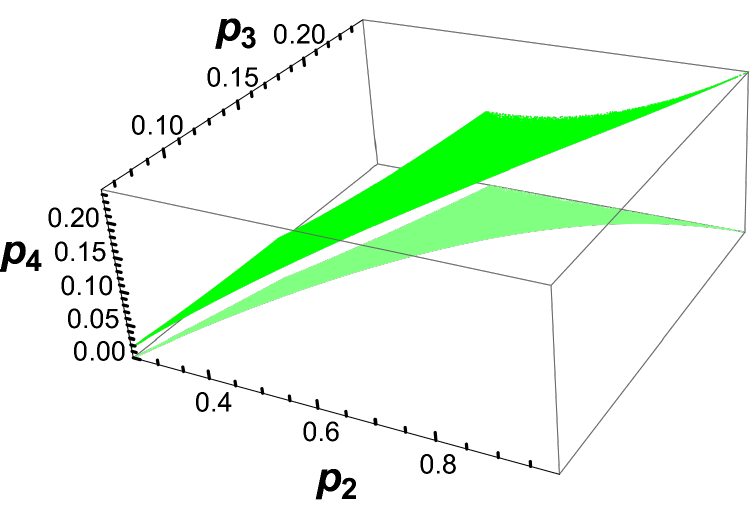}
\end{minipage}}
\subfigure[]{\begin{minipage}[b]{0.45\linewidth}
\includegraphics[width=1\textwidth]{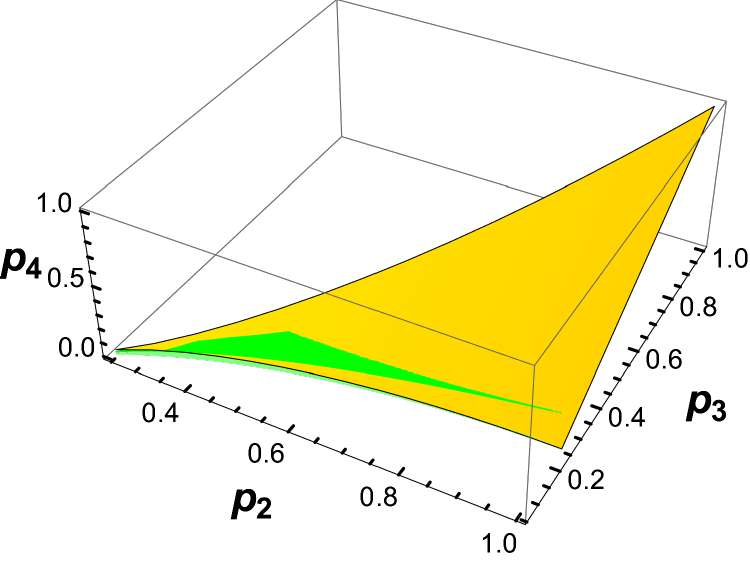}
\end{minipage}}
\caption{(Color Online) (a) The upper region in green comprised of
$(p_2,p_3,p_4)$ corresponds to the family of Bell-diagonal states, and the lower region in light green is 2D comprised of $(p_2,p_3)$ and is the projection of the upper region; (b) The dividing
yellow surface represented by $p_4=F(p_2,p_3)$ crosses the green region of $(p_2,p_3,p_4)$. Those
points over the surface $p_4=F(p_2,p_3)$ correspond to entangled Bell-diagonal
states; those ones under the surface $p_4=F(p_2,p_3)$ correspond to separable Bell-diagonal
states. }\label{fig:bell3d}
\end{figure}
We also distinguish the separable region vs entangled region for
the family of Bell-diagonal states, see
Figure~\ref{fig:bellsepvsent}.
\begin{figure}[h!]\centering
{\begin{minipage}[b]{0.6\linewidth}
\includegraphics[width=1\textwidth]{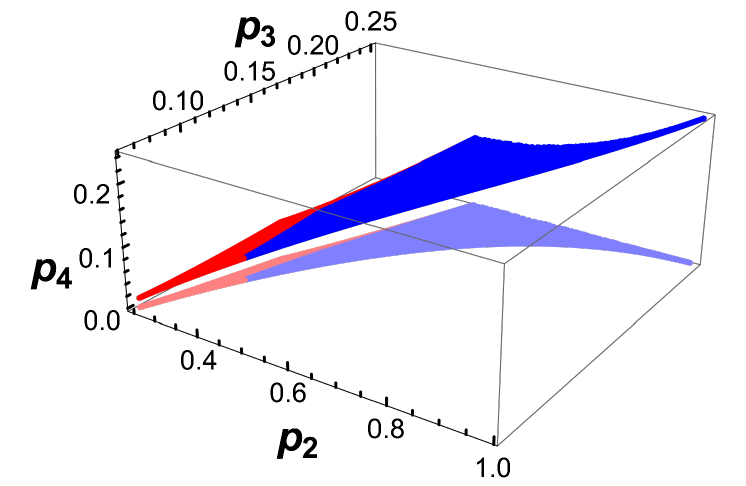}
\end{minipage}}
\caption{(Color Online) For the family of Bell-diagonal states, the red (blue) 3D region comprised of $(p_2,p_3,p_4)$
corresponds to separable (entangled) states with their projected regions comprised of
$(p_2,p_3)$ marked in light red (blue). }\label{fig:bellsepvsent}
\end{figure}
\end{exam}

\section{Visualization of the 3D region of all two-qubit states}\label{sec:fig}

As shown in Figure~\ref{fig:werner3d}, Figure~\ref{fig:bell3d} and
Figure~\ref{fig:bellsepvsent}, we have visualized the 3D regions
corresponding to the family of Werner states and that of
Bell-diagonal states. In this section we further visualize the 3D
region in the $(p_2,p_3,p_4)$-coordinate system corresponding to the
set of all two-qubit states in Figure \ref{fig:2qubitsepvsent},
according to Theorem \ref{thm:main}. Based on Proposition
\ref{prop:ent}, we identify the separable region (corresponding to
the set of separable states) and the entangled region (corresponding
to the set of entangled states), and mark such two regions in red
and blue respectively in Figure~\ref{fig:2qubitsepvsent} (a). We
also add the yellow surface represented by $p_4=F(p_2,p_3)$ in
Figure \ref{fig:2qubitsepvsent} (b), which crosses the 3D region and
divides the separable and entangled regions.

Finally, we note that the 3D region comprised of $(p_2,p_3,p_4)$ in
Figure~\ref{fig:2qubitsepvsent} is bounded by the two surfaces $S_+$
and $S_-$ defined below:
\begin{eqnarray}\label{eq:vs-1}
S_+ &: & \quad p_4 = F(p_2,p_3) + \frac14,  \\
S_-&: & \quad p_4 = F(p_2,p_3) - \frac{1}{64}.
\end{eqnarray}
It follows from Eq.~\eqref{eq:mMbound} that $m(p_2,p_3)$ is lower
bounded by $-\frac{1}{16}$ and $M(p_2,p_3)$ is upper bounded by
$\frac{1}{256}$. Substituting such two bounds into Eqs.
\eqref{eq:F+todet-} and \eqref{eq:F-todet+} respectively, we
conclude that
\begin{equation}
\label{eq:S+-} F(p_2,p_3) - \frac{1}{64}\leqslant
F^-(p_2,p_3)\leqslant F^+(p_2,p_3)\leqslant F(p_2,p_3) + \frac14.
\end{equation}
Graphically, the 3D region comprised of $(p_2,p_3,p_4)$
corresponding to the set of all two-qubit states is lower bounded by
$S_-$ and upper bounded by $S_+$.
\begin{figure}[h!]
\subfigure[]{\begin{minipage}[b]{0.45\linewidth}
\includegraphics[width=1\textwidth]{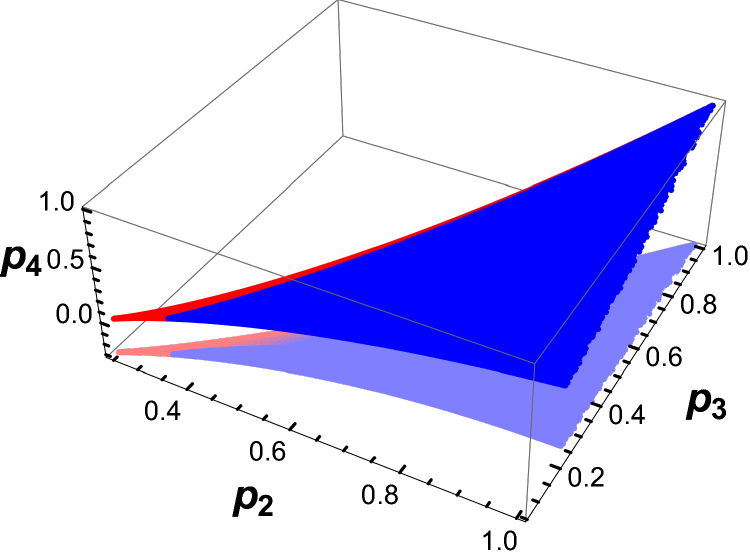}
\end{minipage}}
\subfigure[]{\begin{minipage}[b]{0.45\linewidth}
\includegraphics[width=1\textwidth]{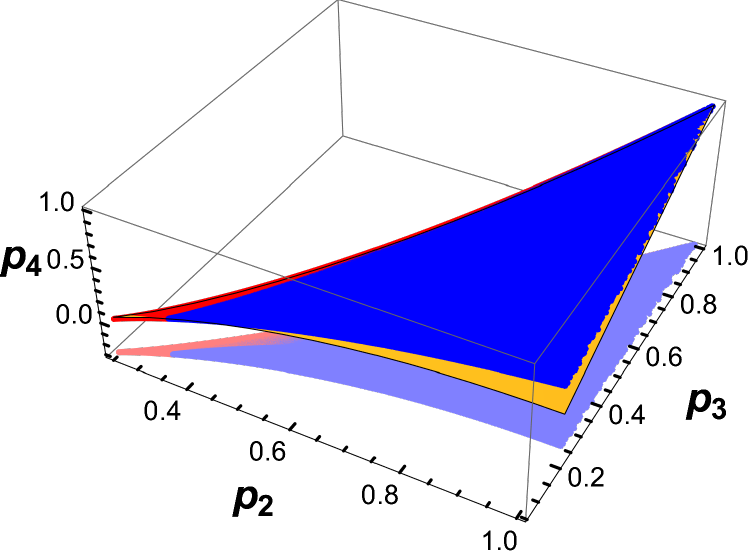}
\end{minipage}}
\caption{(Color Online) (a) The red (blue) 3D region of
$(p_2,p_3,p_4)$ corresponds to the set of all two-qubit separable
(entangled) states with their projections of $(p_2,p_3)$ marked in light red (blue) on the plane
$p_4=-0.3$; (b) The dividing surface $p_4=F(p_2,p_3)$ in yellow
crosses the 3D region of $(p_2,p_3,p_4)$.}\label{fig:2qubitsepvsent}
\end{figure}

\section{Comparison between PT-moments and some entanglement measures on the entanglement detection}\label{sec:com}

Negativity and concurrence are two common entanglement measures which have been widely investigated. For any entanglement measure $E$, $E(\rho_{AB})$ is equal to zero if the state $\rho_{AB}$ is separable, and $E(\rho_{AB})$ is greater than zero if the state $\rho_{AB}$ is entangled. Thus, an analytically computable entanglement measure is also an effective tool to detect entanglement. In this section we shall compare PT-moments and concurrence in the aspect of detecting two-qubit entanglement, via the computable entanglement measure, i.e. the negativity.

For two-qubit states, it is known from Ref. \cite{Verstraete2001}
that the negativity (denoted by $\cN$) and the
concurrence (denoted by $C$) have the following relation:
\begin{equation}
\label{eq:NandC-1}
  \sqrt{(1-C)^2+C^2}-(1-C)\leqslant \cN\leqslant C.
\end{equation}
According to this relation, one can verify that $\cN(\rho_{AB})=0$ if and only if $C(\rho_{AB})=0$, and $\cN(\rho_{AB})>0$ if and only if $C(\rho_{AB})>0$. Moreover, due to the PPT criterion on two-qubit states, we conclude that $\cN(\rho_{AB})=0$ if and only if $\rho_{AB}$ is separable. Therefore, using the concurrence to detect two-qubit entanglement is as powerful as using the negativity, although the values of the two measures on some entangled state could be different. Based on this fact, we may use the negativity as an intermediate parameter to discuss the relation between concurrence and PT-moments. Both the negativity and the PT-moment are defined by the eigenvalues of $\rho^\Gamma$. Here we assume that $\lambda_1,\ldots,\lambda_4$ are the four eigenvalues of $\rho^\Gamma$ in the descending order. Denote the elementary symmetric polynomials of $\lambda_1,\ldots,\lambda_4$ by
$e_1,\ldots,e_4$, i.e.
\begin{equation}
  \label{eq:esp-1}
 e_1=\sum^4_{k=1}\lambda_k,~~ e_2=\sum_{1\leqslant i<j\leqslant
 4}\lambda_i\lambda_j,~~ e_3=\sum_{1\leqslant i<j<k\leqslant
 4}\lambda_i\lambda_j\lambda_k,~~e_4=\prod^4_{k=1}\lambda_k.
\end{equation}
If we get measured all PT-moments $p_1,p_2,p_3,p_4$ with the first PT-moment fixed by constant $1$, then we can employ the following formulae \cite{Duan2022}
\begin{eqnarray}
\label{eq:esp-2}
e_k=\frac1{k!}\Abs{\begin{array}{ccccc}
                     p_1 & 1 & 0 & \cdots & 0 \\
                     p_2 & p_1 & 2 & \cdots & 0 \\
                     \vdots & \vdots & \vdots & \ddots & \vdots \\
                     p_{k-1} & p_{k-2} & p_{k-3} & \cdots & k-1 \\
                     p_k & p_{k-1} & p_{k-2} & \cdots & p_1
                   \end{array}}(k\geqslant1)
\end{eqnarray}
to calculate the value for each $e_k$. Then we construct the $4$th-order
polynomial equation as
\begin{eqnarray}
\label{eq:esp-3}
x^4-e_1 x^3+e_2x^2-e_3x+e_4=0.
\end{eqnarray}
According to the relationship between roots and coefficients, the
four roots of Eq. \eqref{eq:esp-3} are $\lambda_1,\cdots,\lambda_4$,
and thus we can numerically calculate $\lambda_1,\cdots,\lambda_4$
for measured $p_1,p_2,p_3,p_4$. For two-qubit state $\rho_{AB}$ we
have $\cN(\rho_{AB})=\max(0,-2\lambda_{\min})$. When $p_2,p_3,p_4$
satisfy $F^-(p_2,p_3)\leqslant p_4\leqslant F(p_2,p_3)$ for
$(p_2,p_3)\in\cA$, the state $\rho_{AB}$ is separable, and thus all
eigenvalues $\lambda_k$ are positive, which leads to
$\cN(\rho_{AB})=0$. When $p_2,p_3,p_4$ satisfy $F(p_2,p_3)<
p_4\leqslant F^+(p_2,p_3)$ for $(p_2,p_3)\in\cA$, the state
$\rho_{AB}$ is entangled, and thus there is only one negative
eigenvalue namely $\lambda_4$, which leads to $\cN(\rho_{AB})>0$.
Based on the equivalence between the negativity and the concurrence
in the aspect of detecting two-qubit entanglement, we finally obtain
the specific relation between PT-moments and the concurrence as
follows.
\begin{itemize}
\item $F^-(p_2,p_3)\leqslant p_4\leqslant F(p_2,p_3)$ for
$(p_2,p_3)\in\cA$ if and only if $C=0$;
\item $F(p_2,p_3)< p_4\leqslant F^+(p_2,p_3)$ for $(p_2,p_3)\in\cA$
if and only if $C>0$.
\end{itemize}
This tells us that the sign of $p_4-F(p_2,p_3)$ can play the role of an entanglement witness as the concurrence. If the sign is $+1$, it means the two-qubit state is entangled. Otherwise, it means the two-qubit state is separable.
Although the criterion via PT-moments has the same performance as the concurrence and equivalently the negativity in the aspect of detecting two-qubit entanglement, there is a unique advantage of using PT-moments. That is, both the separable and entangled regions can be visualized in the $(p_2,p_3,p_4)$-coordinate. This is also a major motivation of this work.

\section{Conclusion}\label{sec:con}

In this paper, we investigated the characterization of two-qubit
states using all PT-moments. Based on the relation between the
second and third PT-moments, i.e., $p_2$ and $p_3$ derived in Ref.
\cite{Zhang2022}, we further involved the last PT-moment $p_4$ for
two-qubit states, and characterized the relation between $p_4$ and
the pair of PT-moments $(p_2,p_3)$. It is a graphical
characterization of all two-qubit states by plotting the 3D region
comprised of $(p_2,p_3,p_4)$. Specifically, $p_4$ can be bounded by
two functions in terms of two variables $p_2$ and $p_3$. We
formulated the lower bound function $F^-(p_2,p_3)$ and the upper
bound function $F^+(p_2,p_3)$ which can be numerically calculated
for a given pair $(p_2,p_3)$. Furthermore, we also identified the
function $F(p_2,p_3)$ which divides the two regions corresponding to
the set of separable states and the set of entangled states
respectively. Due to the measurability of PT-moments, we provided an
operational criterion for the detection of two-qubit entanglement.
To clarify how to obtain $F^-(p_2,p_3)$ and $F^+(p_2,p_3)$, we
summarized the calculating procedure, and calculated the two bound
functions for the family of Werner states and the family of
Bell-diagonal states following the calculating procedure. For such
two families of states, we plotted the 3D regions accordingly, and
marked the dividing surfaces. Generally, we further visualized the 3D region
in the $(p_2,p_3,p_4)$-coordinate system corresponding to the set of
all two-qubit states, and marked the separable region and entangled
region in different colors. Finally, we concluded that the criterion we proposed in terms of the triple $(p_2,p_3,p_4)$ has the same performace as the two common measures, i.e. the negativity and the concurrence, in the aspect of detecting two-qubit entanglement.

\subsubsection*{Acknowledgements}
This research is supported by Zhejiang Provincial Natural Science Foundation of China under Grant No. LZ23A010005 and by NSFC under Grant No.11971140.
YS is supported by the Fundamental Research Funds for the Centeral Universities under Grant No. JUSRP123029.



\newpage

\appendix
\appendixpage
\addappheadtotoc

\section{The rough conclusion on $s_{\max}$ and $s_{\min}$}
\label{sec:appsmaxmin}

\begin{lem}
\label{le:smaxmin}
The maximum and minimum values $s_{\max}$ and $s_{\min}$ defined by Eqs. \eqref{eq:smax} and \eqref{eq:smin} respectively are reached only if at least two of three variables in $s(x,y,z)$ are equal.
\end{lem}

\begin{proof}
By using Lagrange Multiplier Method, we define the Lagrange function
$L(x,y,z;\mu,\nu)$ as
\begin{eqnarray}
L(x,y,z;\mu,\nu)&=&
x+y+z+\mu\Pa{x^2+y^2+z^2+(1-x-y-z)^2-p_2}\notag\\
&&~~~~~~~~~~~~~~+\nu\Pa{x^3+y^3+z^3+(1-x-y-z)^3-p_3}.
\end{eqnarray}
In order to get the stationary points, we let all partial
derivatives be zero. Then it follows that
\begin{eqnarray}
0=\frac{\partial L}{\partial x} &=& 2\mu(2x+y+z-1)-3\nu(y+z-1)(2x+y+z-1)+1\label{eq:x},\\
0=\frac{\partial L}{\partial y} &=& 2\mu(x+2y+z-1)-3\nu(x+z-1)(x+2y+z-1)+1\label{eq:y},\\
0=\frac{\partial L}{\partial z} &=& 2\mu(x+y+2z-1)-3\nu(x+y-1)(x+y+2z-1)+1\label{eq:z},\\
0=\frac{\partial L}{\partial \mu} &=&x^2+y^2+z^2+(1-x-y-z)^2-p_2\label{eq:p2},\\
0=\frac{\partial L}{\partial \nu}
&=&x^3+y^3+z^3+(1-x-y-z)^3-p_3.\label{eq:p3}
\end{eqnarray}
The three differences given by \eqref{eq:x}$-$\eqref{eq:y},
\eqref{eq:x}$-$\eqref{eq:z}, and \eqref{eq:y}$-$\eqref{eq:z} lead to:
\begin{eqnarray}
(x-y)[3\nu(x+y)+2\mu]=(x-z)[3\nu(x+z)+2\mu]=(y-z)[3\nu(y+z)+2\mu]=0.
\end{eqnarray}
It follows from the above chain of equations that the stationary points satisfy that at least two of three variables $x,y,z$ are equal.
\end{proof}

\section{The global maximal and minimal values of $\det(\rho_{AB}^\Gamma)$ over all two-qubit states}
\label{sec:gdet}

For the global maximal
value, we obtain that
\begin{equation}
\label{eq:detgmax}
\begin{aligned}
&\max\Set{\det(\rho^\Gamma_{AB}):\rho_{AB}\in\density{\complex^2\ot\complex^2}}\\
&= \max \Set{xyz(1-x-y-z):x\geqslant y\geqslant z\geqslant
1-x-y-z\geqslant0}\\
&=\frac1{256},
\end{aligned}
\end{equation}
which is attained only if $x=y=z=\frac14$. For the global minimal value, we obtain that
\begin{equation}
\label{eq:detgmin}
\begin{aligned}
&\min\Set{\det(\rho^\Gamma_{AB}):\rho_{AB}\in\density{\complex^2\ot\complex^2}}\\
&= \min \Set{xyz(1-x-y-z):x\geqslant y\geqslant z>0,~0>1-(x+y+z)\geqslant-\frac12}\\
&= \min \Set{xyz(1-x-y-z):x\geqslant y\geqslant z>0,~\frac32\geqslant x+y+z>1}\\
&= - \max\Set{xyz(x+y+z-1):x\geqslant y\geqslant z>0,
\frac32\geqslant x+y+z>1}\\
&=-\frac1{16},
\end{aligned}
\end{equation}
which is attained only if $x=y=z=\frac12$. It implies that $\det(\rho^\Gamma_{AB})$ attains its global minimal value if and only if $\rho_{AB}$ is maximally entangled. Now we conclude that
\begin{eqnarray}
-\frac1{16}\leqslant \det(\rho^\Gamma_{AB})\leqslant
\frac1{256},~\forall \rho_{AB}\in
\density{\complex^2\ot\complex^2}.
\end{eqnarray}

\section{The derivation process of Eqs. \eqref{eq:MtoP} and \eqref{eq:mtoP}}
\label{sec:MmtoP}

Define
$\Phi(x,y,z):=\det(\rho^\Gamma_{AB})=xyz(1-x-y-z)$, where $x,y,z$
are the largest three eigenvalues of $\rho^\Gamma_{AB}$. For any
fixed pair $(p_2,p_3)\in\cA$, the three variables necessarily meet
all the constraints given by Eqs.~\eqref{eq:maincons-1} ---
\eqref{eq:maincons-4}. Recall the notation $\cV(p_2,p_3)$ introduced
below those constraints in Sect.~\ref{sec:pre}. Then we have
$(x,y,z)\in\cV(p_2,p_3)$. Now our optimization problem is formulated
as
\begin{eqnarray} \label{eq:detformula} \min/\max
~~\{\Phi(x,y,z):~(x,y,z)\in\cV(p_2,p_3)\}.
\end{eqnarray}
To figure out this optimization problem, let $s:=x+y+z$ be fixed.
Using the notations defined by Eqs.~\eqref{eq:smax} ---
\eqref{eq:smin}, we obtain that
$s\in[s_{\min},s_{\max}]\subseteq[\frac34,\frac32]$. Then the
objective function to be optimized becomes
$\Phi(x,y,z;s):=(1-s)xyz$. The problem of optimizing $\Phi(x,y,z)$
is equivalent to optimizing the objective function $\Phi(x,y,z;s)$
with respect to $s\in[s_{\min},s_{\max}]$ after the optimization of
$\Phi(x,y,z;s)$ with the triple of variables $(x,y,z)$ subject to
the following constraints:
\begin{eqnarray} \label{eq:cons-0-s2}
&&x+y+z=s,\\
\label{eq:cons-1-s2}
&&1\geqslant x\geqslant y\geqslant z\geqslant 0,\\
\label{eq:cons-2-s2}
&&z\geqslant 1-s \geqslant -\frac12,\\
\label{eq:cons-3-s2}
&&x^2+y^2+z^2=p_2-(1-s)^2,\\
\label{eq:cons-4-s2}
&&x^3+y^3+z^3=p_3-(1-s)^3.
\end{eqnarray}

One can verify the following observation:
\begin{eqnarray}
\label{eq:xyz-s}
xyz = \frac1{3!}\Abs{\begin{array}{ccc}
                        s & 1 & 0 \\
                        p_2-(1-s)^2 & s & 2 \\
                        p_3-(1-s)^3 & p_2-(1-s)^2 & s
                      \end{array}
}.
\end{eqnarray}
Based on the above observation, we obtain that
\begin{equation}
\label{eq:phi-p2-p3}
\begin{aligned}
\Phi(x,y,z;s) &= (1-s)xyz=(1-s)\Pa{s^3-2s^2+\frac{3-p_2}2s+\frac{p_3-1}3}\\
&=-\Pa{s^4-3s^3+\frac{7-p_2}2s^2+\frac{3p_2+2p_3-11}6s+\frac{1-p_3}3} \\
&\equiv \Phi(s|p_2,p_3).
\end{aligned}
\end{equation}
Here, we use $\Phi(s|p_2,p_3)$ to represent the second line of Eq.
\eqref{eq:phi-p2-p3}. For convenience, let $P(s|p_2,p_3)=-\Phi(s|p_2,p_3)$, and then the target problem is equivalently transformed to optimizing $P(s|p_2,p_3)$ over the closed interval
$[s_{\min},s_{\max}](\subseteq[\tfrac34,\frac32])$ for any given pair $(p_2,p_3)\in\cA$.

\section{The optimization process of $P(s|p_2,p_3)$}
\label{sec:optP}

We calculate the first and second derivatives
of $P(s|p_2,p_3)$ with respect to $s$ as
\begin{eqnarray}
\frac{\dif}{\dif s}P(s|p_2,p_3) &=& 4s^3 - 9s^2 + (7-p_2)s + \frac{3p_2+2p_3-11}6,\label{eq:Pder-1}\\
\frac{\dif^2}{\dif s^2}P(s|p_2,p_3) &=& 12s^2 -
18s+(7-p_2).\label{eq:Pder-2}
\end{eqnarray}
Consider the second derivative $\frac{\dif^2}{\dif s^2}P(s|p_2,p_3)$
which is a quadratic function with respect to $s$. By calculation
the discriminant of the quadratic equation $\frac{\dif^2}{\dif
s^2}P(s|p_2,p_3)=0$ is $48(p_2-\tfrac14)$. Since
$p_2\in[\tfrac14,1]$, the discriminant is always nonnegative. Thus,
the two roots of the equation are
\begin{eqnarray}
s_{\pm}=\frac{9\pm\sqrt{3(4p_2-1)}}{12}.
\end{eqnarray}
One can verify that $s_-\in[\tfrac12, \tfrac34]$ and
$s_+\in[\tfrac34, 1]$ for $p_2\in[\tfrac14,1]$. Due to $s\in
[s_{\min}, s_{\max}]\subseteq[\tfrac34,\frac32]$, we conclude that
$\frac{\dif^2}{\dif s^2}P(s|p_2,p_3)\geqslant0$ when $s\in[s_+,1]$;
and $\frac{\dif^2}{\dif s^2}P(s|p_2,p_3)\leqslant0$ when $s\in
[\tfrac34,s_+]$. Next, we consider the distribution of all roots of
the cubic equation $\frac14\frac{\dif}{\dif s}P(s|p_2,p_3)=0$ with
respect to $s$, namely
\begin{eqnarray}
\label{eq:firstprime} s^3 - \frac94s^2 + \frac{7-p_2}4s +
\frac{3p_2+2p_3-11}{24}=0.
\end{eqnarray}
Taking $s=t+\frac34$ in Eq. \eqref{eq:firstprime} we obtain that
\begin{eqnarray}\label{eq:cubic-1}
t^3 + p t + q =0,
\end{eqnarray}
where
\begin{eqnarray}\label{eq:pvsq}
p = \frac{1-4p_2}{16},\quad q = \frac{1-6p_2+8p_3}{96}.
\end{eqnarray}
Its discriminant is given by $D(p_2,p_3) =
-108\cdot\Delta(p_2,p_3)$, where
\begin{eqnarray}
\Delta(p_2,p_3) &=&\frac{q^2}4+\frac{p^3}{27}= \frac1{576}\Pa{p^2_3
+
\frac{1-6p_2}4p_3 + \frac{1-12p_2+39p^2_2-16p^3_2}{48}}\\
&=&\frac1{576}\Pa{p_3-f^+(p_2)}\Pa{p_3-f^-(p_2)},
\end{eqnarray}
where $f^{\pm}(p_2)$ are formulated in Eq. \eqref{eq:p2-p3-f}. It
follows from Eq. \eqref{eq:p2-p3} that $p_3$ is bounded by
$f^{\pm}(p_2)$ for any two-qubit state. That is, $f^-(p_2)\leqslant
p_3\leqslant f^+(p_2)$ for any fixed pair $(p_2,p_3)\in \cA$. Hence,
we conclude that $\Delta(p_2,p_3)=0$ if $p_3=f^+(p_2)$ or
$p_3=f^-(p_2)$, otherwise $\Delta(p_2,p_3)$ is always negative. We
also plot the surface graph in the $(p_2,p_3,\Delta)$-coordinate
system, see Fig.~\ref{fig:delta}, to perceive the sign of
$\Delta(p_2,p_3)$ more intuitively.
\begin{figure}[ht]\centering
{\begin{minipage}[b]{1\linewidth}
\includegraphics[width=1\textwidth]{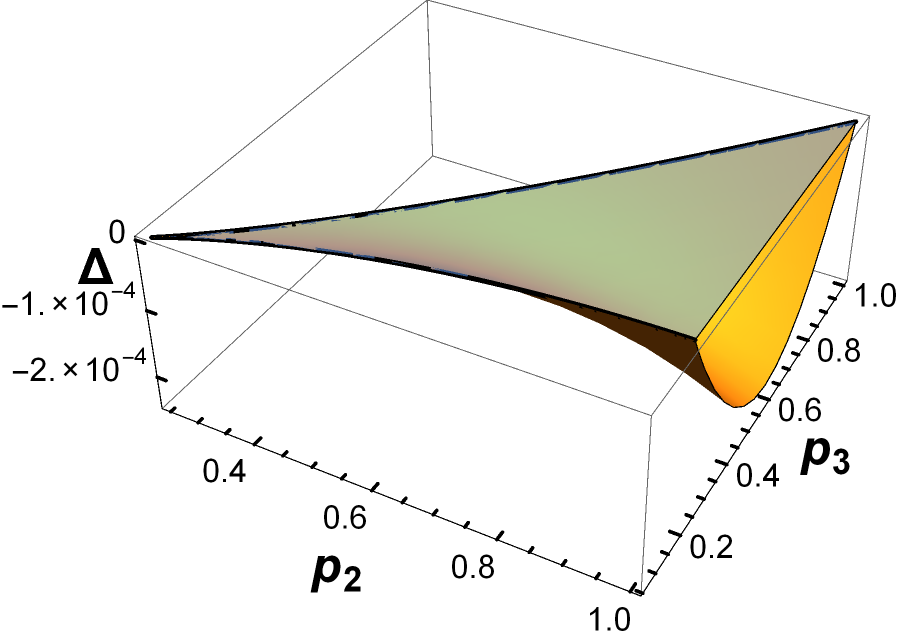}
\end{minipage}}
\caption{(Color Online) The surface corresponding to
$\Delta(p_2,p_3)$ supported on the 2D region $\cA$. The surface below the plane of $\Delta(p_2,p_3)=0$ except the two boundary curves $p_3=f^\pm(p_2)$. The two curves $p_3=f^\pm(p_2)$ both lie in the plane of $\Delta(p_2,p_3)=0$.
}
\label{fig:delta}
\end{figure}
Based on the sign of $\Delta(p_2,p_3)$, we shall investigate the
distribution of all roots of the equation $\frac14\frac{\dif}{\dif
s}P(s|p_2,p_3)=0$ in two cases. The first case is when
$\Delta(p_2,p_3)=0$, and the second is when $\Delta(p_2,p_3)<0$.
Note that $(p_2,p_3)\in\cA$, then $p_2=\frac14$ implies that
$p_3=\frac1{16}$. Then it must be true that $p=0$ and $q=0$ at the
same time. In a word, beyond the boundary of $\cA$, we see that
\begin{eqnarray}
p<0\quad\text{and}\quad \Delta<0.
\end{eqnarray}
In addition, $q$ vanished only when the point $(p_2,p_3)$ lies on
the line $p_3=f(p_2)=\frac{6p_2-1}8$.

\begin{enumerate}[(1)]
\item If $\Delta(p_2,p_3)=0$, then $p_3=f^\pm(p_2)$ for $p_2\in[\tfrac14,1]$.
According to Appendix~\ref{app:cubicequation}, the three roots of
the equation $\frac14\frac{\dif}{\dif s}P(s|p_2,p_3)=0$ are all
real. We further consider the following two subcases.
\begin{itemize}
\item[(1a)] If $p_3=f^-(p_2)$, then $q<0$ since $f(p_2)>f^-(p_2)$, the three
roots of the equation $\frac14\frac{\dif}{\dif s}P(s|p_2,p_3)=0$
comprise of one simple root and double root, respectively, which are
given by the formulae:
\begin{eqnarray}
\ell_+:=\frac34+\frac{3q}{p} = \frac{9+2\sqrt{3(4p_2-1)}}{12}\text{
and }\ell_-:=\frac34-\frac{3q}{2p} = \frac{9-\sqrt{3(4p_2-1)}}{12}.
\end{eqnarray}
It implies that
\begin{eqnarray}
\frac14\frac{\dif}{\dif s}P(s|p_2,p_3)=
\Pa{s-\ell_+}\Pa{s-\ell_-}^2,
\end{eqnarray}
and all the three real roots are $s_1=\ell_+$ and $s_2=s_3=\ell_-$.
One can verify only the root
$s_1=\frac{9+2\sqrt{3(4p_2-1)}}{12}\in[\tfrac34,\frac32]$. Hence, we
only need to consider the stationary point with $s=\ell_+$. If
$\ell_+\in[s_{\min},s_{\max}]$, then
\begin{eqnarray*}
\min_{s\in[s_{\min},s_{\max}]}P(s|p_2,p_3) =
\min\Set{P(\ell_+|p_2,f^-(p_2)),P(s_{\max}|p_2,f^-(p_2)),P(s_{\min}|p_2,f^-(p_2))},\\
\max_{s\in[s_{\min},s_{\max}]}P(s|p_2,p_3) =
\max\Set{P(\ell_+|p_2,f^-(p_2)),P(s_{\max}|p_2,f^-(p_2)),P(s_{\min}|p_2,f^-(p_2))}.
\end{eqnarray*}
Otherwise, if $\ell_+\notin[s_{\min},s_{\max}]$, then
\begin{eqnarray*}
\min_{s\in[s_{\min},s_{\max}]}P(s|p_2,p_3) =
\min\Set{P(s_{\max}|p_2,f^-(p_2)),P(s_{\min}|p_2,f^-(p_2))},\\
\max_{s\in[s_{\min},s_{\max}]}P(s|p_2,p_3) =
\max\Set{P(s_{\max}|p_2,f^-(p_2)),P(s_{\min}|p_2,f^-(p_2))}.
\end{eqnarray*}
\item[(1b)] If $p_3=f^+(p_2)$, then $q>0$ since $f^+(p_2)>f(p_2)$, the three
roots of the equation $\frac14\frac{\dif}{\dif s}P(s|p_2,p_3)=0$
comprise of one simple root and double root, respectively, which are
given by the formulae:
\begin{eqnarray}
L_-:=\frac34+\frac{3q}{p} = \frac{9-2\sqrt{3(4p_2-1)}}{12}\text{ and
}L_+:=\frac34-\frac{3q}{2p} = \frac{9+\sqrt{3(4p_2-1)}}{12}.
\end{eqnarray}
We see that
\begin{eqnarray}
\frac14\frac{\dif}{\dif s}P(s|p_2,p_3)= \Pa{s-L_-}\Pa{s-L_+}^2,
\end{eqnarray}
are the three roots are $s_1=s_2=L_+$ and $s_3=L_-$. One can verify
that only the double roots $s_1=s_2=\frac{9+\sqrt{3(4p_2-1)}}{12}\in
[\tfrac34,1]$. Hence, we only need to consider the stationary point
with $s=L_+$. If $L_+\in[s_{\min},s_{\max}]$, then
\begin{eqnarray*}
\min_{s\in[s_{\min},s_{\max}]}P(s|p_2,p_3) =
\min\Set{P(L_+|p_2,f^+(p_2)),P(s_{\max}|p_2,f^+(p_2)),P(s_{\min}|p_2,f^+(p_2))},\\
\max_{s\in[s_{\min},s_{\max}]}P(s|p_2,p_3) =
\max\Set{P(L_+|p_2,f^+(p_2)),P(s_{\max}|p_2,f^+(p_2)),P(s_{\min}|p_2,f^+(p_2))}.
\end{eqnarray*}
If $L_+\notin[s_{\min},s_{\max}]$, then
\begin{eqnarray*}
\min_{s\in[s_{\min},s_{\max}]}P(s|p_2,p_3) =
\min\Set{P(s_{\max}|p_2,f^+(p_2)),P(s_{\min}|p_2,f^+(p_2))},\\
\max_{s\in[s_{\min},s_{\max}]}P(s|p_2,p_3) =
\max\Set{P(s_{\max}|p_2,f^+(p_2)),P(s_{\min}|p_2,f^+(p_2))}.
\end{eqnarray*}
\end{itemize}
\item If $\Delta(p_2,p_3)<0$, then $f^-(p_2)<p_3<f^+(p_2)$ for $\frac 14<p_2\leqslant 1$. According to Appendix~\ref{app:cubicequation}, the cubic equation formulated in Eq. \eqref{eq:cubic-1}, with coefficients given by Eq. \eqref{eq:pvsq},
has three real roots given by
\begin{eqnarray}
\label{eq:Delta<0roots-1}
\begin{cases}
t_1 = \frac{\sqrt{3(4p_2-1)}}6\cos\theta\\
t_2 = \frac{\sqrt{3(4p_2-1)}}6\cos\Pa{\theta-\frac{2\pi}3}\\
t_3 = \frac{\sqrt{3(4p_2-1)}}6\cos\Pa{\theta+\frac{2\pi}3}
\end{cases},
\end{eqnarray}
where
\begin{eqnarray}
\label{eq:Delta<0roots-2}
\theta\equiv\theta(p_2,p_3)=\frac13\arccos\Pa{\frac{\sqrt{3}(6p_2-8p_3-1)}{(4p_2-1)^\frac32}}.
\end{eqnarray}
Recall that $\tfrac14\frac{\dif}{\dif s}P(s|p_2,p_3)=0$ can be
transformed to Eq.~\eqref{eq:cubic-1} by taking $s=t+\frac 34$. It
follows that $\tfrac14\frac{\dif}{\dif s}P(s|p_2,p_3)=0$ has three
real roots given by
\begin{eqnarray}
\label{eq:Delta<0roots-3}
\begin{cases}
s_1=r_1(p_2,p_3):=\frac34+\frac{\sqrt{3(4p_2-1)}}6\cos\theta\\
s_2=r_2(p_2,p_3):=\frac34+\frac{\sqrt{3(4p_2-1)}}6\cos\Pa{\theta-\frac{2\pi}3}\\
s_3=r_3(p_2,p_3):=\frac34+\frac{\sqrt{3(4p_2-1)}}6\cos\Pa{\theta+\frac{2\pi}3}
\end{cases}.
\end{eqnarray}
The local extreme values can be attained only at the stationary
points and boundary points. In order to roughly determine the values
of these roots for a fixed pair $(p_2,p_3)\in\cA$, we plot the 3D
surfaces of the above three roots $s_i=r_i(p_2,p_3)$, where
$i=1,2,3$, in terms of $p_2$ and $p_3$, see Fig. \ref{fig:3roots}.
\begin{figure}[ht]\centering
{\begin{minipage}[b]{1\linewidth}
\includegraphics[width=1\textwidth]{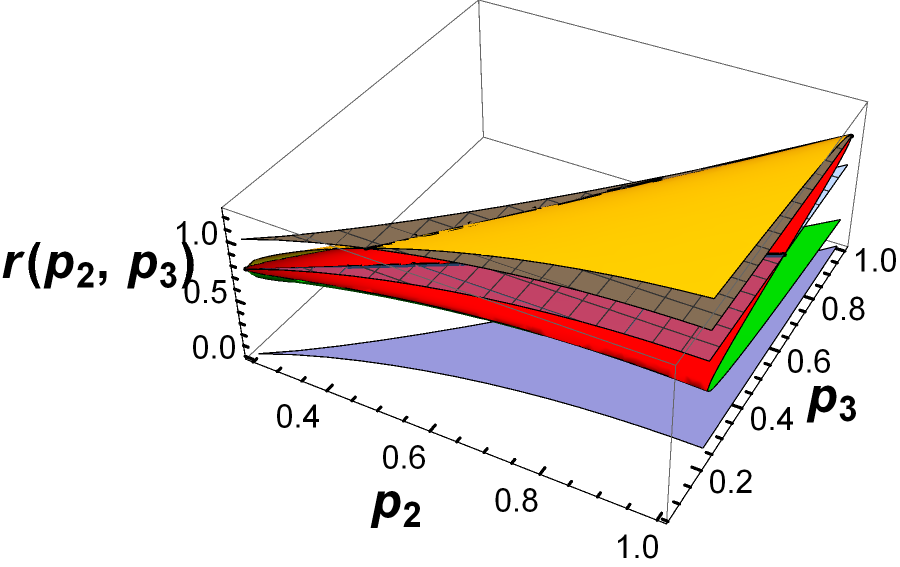}
\end{minipage}}
\caption{(Color Online) The three distinct real roots $r(p_2,p_3)$
of the cubic equation $\frac14\frac{\dif}{\dif s}P(s|p_2,p_3)=0$ for
any pair $(p_2,p_3)\in\cA$. The root $r(p_2,p_3)=1$ stands for the line
$p_3=\frac{3p_2-1}2$ with $p_2\in[\frac7{16},1]$, and the root
$r(p_2,p_3)=\frac34$ stands for the line $p_3=\frac{6p_2-1}8$ with
$p_2\in[\frac14,1]$.}\label{fig:3roots}
\end{figure}
\end{enumerate}

Combining the two cases investigated above, we find that for any
pair $(p_2,p_3)\in\cA$, the three roots of $\tfrac14\frac{\dif}{\dif
s}P(s|p_2,p_3)=0$ in the descending order, i.e. $r_1\geqslant
r_2\geqslant r_3$, satisfy that
\begin{eqnarray} \frac14\leqslant
r_3\leqslant\frac34\leqslant r_1\leqslant\frac54\quad\text{and}\quad
\frac12\leqslant r_2\leqslant1.
\end{eqnarray}
In particular, according to the first case when $p_3=f^{\pm}(p_2)$, we determine that $r_1=r_2\in[\frac34,1]$ on the upper boundary curve of $\cA$, i.e. $p_3=f^+(p_2)$, and $r_2=r_3\in[\frac12,\frac34]$ on the lower boundary curve of $\cA$, i.e. $p_3=f^-(p_2)$.
In addition, we find that $r_1(p_2,p_3)=1$ determines a curve $p_3=\frac{3p_2-1}2$ in the 2D region of $\cA$,
where $p_2\in[\frac7{16},1]$, and
$r_2(p_2,p_3)=\frac34$ determines a curve $p_3=\frac{6p_2-1}8$,
where $p_2\in[\frac14,1]$.

\section{Cardano formula}\label{app:cubicequation}

The Cardano formula is named after G. Cardano, who was the first to
publish it in 1545. This formula is one for finding the roots of the
general cubic equation over $\complex$. In our paper, we are
interested in those cubic equations whose all roots are real. For
such class of cubic equations to be used in optimization, Bauschke,
Lal, and Wang give a thorough consideration in their recent paper
\cite{Bauschke2023}:

\begin{thrm}
Let $f(x):=x^3+px+q$, where both $p$ and $q$ are in $\real$. Then
$0$ is the only inflection point of $f$: $f$ is strictly concave on
$(-\infty,0)$ and $f$ is strictly convex on $(0,+\infty)$. Moreover,
exactly one of the following occurs:
\begin{enumerate}[(i)]
\item \blue{$p<0$:} Set $x_\pm:=\sqrt{-\frac p3}$. Then $x_-<x_+$,
where $x_\pm$ are two distinct simple roots of the equation
$f'(x)=0$,
\begin{itemize}
\item $f$ is strictly increasing on $(-\infty,x_-]$;
\item $f$ is strictly decreasing on $[x_-,x_+]$;
\item $f$ is strictly increasing on $[x_+,+\infty)$.
\end{itemize}
Moreover,
\begin{eqnarray}
f(x_+)f(x_-)=4\Delta,\quad \Delta:=\Pa{\frac p3}^3+\Pa{\frac q2}^2,
\end{eqnarray}
and this case trifurcates further as follows:
\begin{enumerate}[(a)]
\item \blue{$\Delta>0$:} Then the equation $f(x)=0$ has \textbf{exactly one real root} $r$. It
is simple and given by
\begin{eqnarray}
r:=u_-+u_+,\quad u_\pm:=\Pa{-\frac q2\pm\sqrt{\Delta}}^{\frac13}.
\end{eqnarray}
The remaining \textbf{two simple nonreal roots} are given by
\begin{eqnarray}
\frac{-(u_-+u_+)\pm\sqrt{-3}(u_--u_+)}2.
\end{eqnarray}
\item \blue{$\Delta=0$:} If $q>0$ ($q<0$), then $2x_-$ ($2x_+$) is a
\textbf{simple real root} while $x_+$ ($x_-$) is a \textbf{double
real root}. Moreover, these cases can be combined into
\begin{itemize}
\item $3\frac qp=2\Pa{-\frac q2}^\frac13$ is a \textbf{simple root} of the equation
$f(x)=0$, and
\item $-\frac{3q}{2p}=-\Pa{-\frac q2}^\frac13$ is a \textbf{double root} of the equation
$f(x)=0$.
\end{itemize}
\item \blue{$\Delta<0$:} Then the equation $f(x)=0$ has \textbf{three simple real roots}:
$r_-,r_0,r_+$, where $r_-<z_-<r_0<z_+<r_+$. Indeed, set
\begin{eqnarray}
\theta:=\arccos\Br{\Pa{-\frac q2}\Pa{-\frac p3}^{-\frac32}},
\end{eqnarray}
which lies in $(0,\pi)$, and then define $x_k(k=0,1,2)$ by
\begin{eqnarray}
x_k:=2\sqrt{-\frac p3}\cos\Pa{\frac{\theta+2k\pi}3}.
\end{eqnarray}
Then $r_-=x_1, r_0=x_2$, and $r_+=x_0$.
\end{enumerate}
\item \blue{$p=0$:} Then the equation $f'(x)=0$ has a \textbf{double root} at
$0$, and $f$ is strictly increasing on $\real$. The only real root
is given by
\begin{eqnarray}
r:=\sqrt[3]{-q}.
\end{eqnarray}
If $q=0$, then $r$ is a triple root. If $q\neq0$, then $r$ is a
simple root and the remaining nonreal simple roots are given by
$\frac{-1\pm\sqrt{-3}}2r$.
\item \blue{$p>0$:} Then the  equation $f'(x)=0$ has no real root,
$f$ is strictly increasing on $\real$, and $f$ has \textbf{exactly
one real root} $r$. It is simple and given by
\begin{eqnarray}
r:=u_-+u_+
\end{eqnarray}
where $u_\pm:=\Pa{-\frac q2\pm\sqrt{\Delta}}^{\frac13}$ and
$\Delta:=\Pa{\frac p3}^3+\Pa{\frac q2}^2$.
\end{enumerate}
Once again, the remaining \textbf{two simple nonreal roots} are
\begin{eqnarray}
r:=u_-+u_+,\quad u_\pm:=\Pa{-\frac q2\pm\sqrt{\Delta}}^{\frac13}.
\end{eqnarray}
\end{thrm}


\begin{prop}\label{prop:diz}
Given a cubic equation with real coefficients:
\begin{eqnarray}\label{eq:3poly}
x^3+a_2x^2+a_1x+a_0=0\quad(a_0\neq0),
\end{eqnarray}
it follows that Eq.~\eqref{eq:3poly} must have at least one real
root. In particular,
\begin{enumerate}[(i)]
\item Let $\gamma$ is a real root of
Eq.~\eqref{eq:3poly}. Then we have
\begin{eqnarray}
\frac{x^3+a_2x^2+a_1x+a_0}{x-\gamma} =
x^2+(a_2+\gamma)x-\frac{a_0}{\gamma}.
\end{eqnarray}
Thus the last two roots are given by
\begin{eqnarray}
\frac{-(a_2+\gamma)\gamma\pm\sqrt{(a_2\gamma +\gamma^2)^2+4
a_0\gamma}}{2\gamma}.
\end{eqnarray}
\item Let $\alpha,\beta$ be two real roots of
Eq.~\eqref{eq:3poly}. Then all roots of Eq.~\eqref{eq:3poly} must be
real, and we have
\begin{eqnarray}
\frac{x^3+a_2x^2+a_1x+a_0}{(x-\alpha)(x-\beta)} = x+a_2+\alpha
+\beta.
\end{eqnarray}
Thus the last real root must be $-(a_2+\alpha+\beta)$.
\end{enumerate}
\end{prop}

\begin{proof}
(i) Note that
\begin{eqnarray*}
\frac{x^3+a_2x^2+a_1x+a_0}{x-\gamma} &=&
x^2+(a_2+\gamma)x+(\gamma^2+a_2\gamma+a_1)+\frac{\gamma^3+a_2\gamma^2+a_1\gamma+a_0}{x-\gamma}\\
&=&x^2+(a_2+\gamma)x+(\gamma^2+a_2\gamma+a_1),
\end{eqnarray*}
where we used the fact that $\gamma$ is the real root of
Eq.~\ref{eq:3poly}, i.e., $\gamma^3+a_2\gamma^2+a_1\gamma+a_0=0$.
Then $\gamma^2+a_2\gamma+a_1=-\frac{a_0}\gamma$ follows from the
fact that $\gamma\neq0$ due to $a_0\neq0$. Therefore, we get that
\begin{eqnarray*}
\frac{x^3+a_2x^2+a_1x+a_0}{x-\gamma} =
x^2+(a_2+\gamma)x-\frac{a_0}{\gamma}.
\end{eqnarray*}
(ii) The proof goes similarly as to (i). Indeed,
\begin{eqnarray*}
\frac{x^3+a_2x^2+a_1x+a_0}{(x-\alpha)(x-\beta)} &=& (x+a_2+\alpha
+\beta)+\frac{\alpha^3+a_2\alpha^2+a_1\alpha+a_0}{(\alpha-\beta)(x-\alpha)}+\frac{\beta^3+a_2\beta^2+a_1\beta+a_0}{(\beta-\alpha)(x-\beta)}\\
&=&x+a_2+\alpha +\beta,
\end{eqnarray*}
which finishes the proof.
\end{proof}

\end{document}